\documentclass[11pt]{article}
\usepackage{amssymb}
\usepackage{colortbl}
\usepackage{amsfonts,amsmath, longtable}

\usepackage{comment}

        \def\theequation{\thesection.\arabic{equation}}

\newcommand{\tr}{{\rm tr}}
\newcommand{\ti}[1]{\tilde{#1}}

\newcommand{\mH}{{\mathcal H}}

\newcommand{\mO}{{\mathcal O}}

\newcommand{\vf}{\varphi}
\newcommand{\al}{\alpha}

\newcommand{\vth}{\vartheta}

\newcommand{\Mat}{ {\rm Mat}(N,\mathbb C) }

\newcommand{\mC}{\mathbb C}
\newcommand{\mZ}{\mathbb Z}

\newcommand{\mS}{\mathcal S}

\newcommand{\z}{{\zeta}}

\newtheorem{predl}{Proposition}

\newenvironment{proof}{\par\noindent{\bf Proof.}}{\hfill$\scriptstyle\blacksquare$}

\def\beq{\begin{equation}}
\def\eq{\end{equation}}
\def\p{\partial}

\newtheorem{theor}{Theorem}

\def\res{\mathop{\hbox{Res}}\limits}

\begin{document}

\setcounter{page}{1}

\begin{center}

\

\vspace{0mm}

{\Large{\bf  On the field analogue of elliptic spin Calogero-Moser
model:}}

\vspace{3mm}

{\Large{\bf
Lax pair and equations of motion
}}



 \vspace{12mm}

 {\Large {Andrei Zotov}}

  \vspace{10mm}

 {\em Steklov Mathematical Institute of Russian
Academy of Sciences,\\ Gubkina str. 8, 119991, Moscow, Russia}


{\em Institute for Theoretical and Mathematical Physics,\\ Lomonosov Moscow State University, Moscow, 119991, Russia}


 {\small\rm {e-mail: zotov@mi-ras.ru}}

\end{center}

\vspace{0mm}

\begin{abstract}
The Lax pair for the field analogue of the classical spin elliptic Calogero-Moser model is proposed.
Namely, using the previously known Lax matrix we suggest an ansatz for the accompany matrix.
The presented construction is valid when the matrix of spin variables ${\mathcal S}\in{\rm Mat}(N,\mathbb C)$
satisfies the condition ${\mathcal S}^2=c_0{\mathcal S}$ with some constant $c_0\in\mathbb C$.
It is proved that the Lax pair satisfies the Zakharov-Shabat equation with unwanted term, thus
providing equations of motion on the unreduced phase space. The unwanted term vanishes
after additional reduction. In the special case ${\rm rank}(\mathcal S)=1$ we show that
the reduction provides the Lax pair of the spinless field Calogero-Moser model obtained earlier
by Akhmetshin, Krichever and Volvovski.
\end{abstract}

%


\bigskip


\section{Introduction}\label{sec1}
\setcounter{equation}{0}
The spinless elliptic Calogero-Moser model \cite{Calogero2} (of $A_{N-1}$ type)
 is an example of $N$-body integrable system of classical finite-dimensional mechanics.
 Dynamics of particles is defined by the equations
of motion
  \beq\label{q02}
  \begin{array}{c}
  \displaystyle{
 {\dot q}_i={\ti b}p_i\,,\qquad  {\dot p}_i={\ti b}{c}^2\sum\limits_{k:k\neq
 i}^N\wp'(q_{i}-q_k)\,,\quad i=1,...,N\,,
 }
 \end{array}
 \eq
 where $p_i$, $q_i$, $i=1,...,N$ are canonically conjugated $\mC$-valued momenta and positions of
 particles\footnote{The Poisson brackets are $\{p_i,q_j\}=\delta_{ij}$, $\{p_i,p_j\}=\{q_i,q_j\}=0$},
 $c\in\mC$ is a coupling constant, $\wp(x)$ is the Weierstrass elliptic $\wp$-function\footnote{See Appendix for elliptic function definitions and properties} and ${\ti b}\neq 0$ is a free $\mC$-valued normalization constant.
 It was shown by I. Krichever \cite{Krich2} that equations (\ref{q02}) are represented in the Lax form
  \beq\label{q01}
  \begin{array}{c}
  \displaystyle{
 \dot{L}(z)+[L(z),M(z)]=0\,,\qquad L(z),M(z)\in\Mat\,,
 }
 \end{array}
 \eq
  where
 $z\in\mC$ is the spectral parameter.
 From (\ref{q01}) it follows that $H_n(z)=\tr(L^n(z))$, $n\in\mZ_+$ are generating functions of integrals of motion.
 Existence of the classical $r$-matrix structure guarantees the involution property $\{H_m(z),H_n(w)\}=0$, $m,n\in\mZ_+$.
 In this way the classical Liouville integrability of the model was established.

 The Calogero-Moser model (\ref{q02}) has the spin generalization \cite{GH,BAB}. In the classical mechanics
 this means existence of the ''spin'' component of the phase space, which is a factor space
  ${\mathcal O}//{\mathcal H}$, where $\mO$ is a coadjoint orbit (of ${\rm GL}_N$ Lie group) and $\mH\subset{\rm GL_N}$ is the Cartan subgroup in ${\rm GL}_N$. Dynamical variables are elements of a matrix $S\in\Mat$. The model is described on the unreduced phase space, that is when
  the spin component of the phase space is $\mO$. This space is equipped with the linear Poisson-Lie structure
  $\{S_{ij},S_{ml}\}=-S_{il}\delta_{mj}+S_{mj}\delta_{il}$, for $i,j,m,l=1,...,N$, and the Casimir functions
  are fixed\footnote{Fixation of the Casimir functions means fixation of eigenvalues of the matrix $S$.}.
  The equations of motion are as follows:
  \beq\label{q03}
  \begin{array}{c}
  \displaystyle{
{\dot q}_i={\ti b}p_i\,,\quad {\dot p}_i={\ti b}\sum\limits_{j:j\neq
i}^NS_{ij}S_{ji}\wp'(q_i-q_j)
 }
 \end{array}
 \eq
and
  \beq\label{q04}
  \begin{array}{c}
  \displaystyle{
{\dot S}_{ii}=0\,,\quad
 {\dot S}_{ij}={\ti b}\sum\limits_{m\neq i,j}^N
S_{im}S_{mj}(\wp(q_i-q_m)-\wp(q_j-q_m))\,,\ i\neq j\,.
 }
 \end{array}
 \eq
  At the same time the Lax equation on the unreduced phase space
  has a certain nontrivial right hand side \cite{BAB}:
  \beq\label{q05}
  \begin{array}{c}
  \displaystyle{
 \dot{L}(z)+[L(z),M(z)]={\rm unwanted\ term}\,.
 }
 \end{array}
 \eq
  The unwanted term vanishes on the reduced phase space. For this purpose
  one should perform  the Poisson reduction
  $\mO\rightarrow {\mathcal O}//{\mathcal H}$. It includes (the first class) constraints
  $\mu_i={\rm const}$, $i=1,...,N$
  \beq\label{q06}
  \begin{array}{c}
  \displaystyle{
 \mu_i=S_{ii}\,,\quad i=1,...,N
 }
 \end{array}
 \eq
together with some (gauge) fixation of the coadjoint action of $\mH$ on $\mO$\footnote{
Conditions (\ref{q06}) arise as moment map equations generated by the coadjoint action of $\mH$ on $\mO$.}.
The reduction changes the Poisson brackets and equations of motion. However it is usually not performed explicitly
since there is no any preferred way to fix the action of $\mH$. Notice also that
the dimension of an orbit $\mO$ depends on the level of fixation of ${\rm Spec}(S)$.
For example, in the minimal orbit case ${\rm Spec}(S)=(c_0,0,...,0)$, that is ${\rm rank}(S)=1$ and
the reduction kills all spin degrees of freedom reproducing the spinless model (\ref{q02}).

In the field case we deal with the field variables. The latter means that the momenta and coordinates of particles
turn into the set of fields $p_1(x),...,p_N(x)$ and $q_1(x),...,q_N(x)$, where $x\in\mathbb R$ is a coordinate on a real line or a circle. For definiteness we assume all the field variables be periodic functions on a circle $\varphi(x)=\varphi(x+2\pi)$. The canonical Poisson structure takes the form
$\{p_i(x),q_j(y)\}=\delta(x-y)\delta_{ij}$, $\{p_i(x),p_j(y)\}=\{q_i(x),q_j(y)\}=0$, $x,y\in\mathbb R$.
The field analogue of the Calogero-Moser model was introduced in \cite{Krich2} and \cite{LOZ} for ${\rm gl}_2$ case and then was extended to ${\rm gl}_N$ case in \cite{Krich22}. Equations of motion are given by the following set of partial differential equations for $q_i(t,x)$, $p_i(t,x)$, $i=1,...,N$:
  \beq\label{q07}
  \begin{array}{c}
  \displaystyle{
 {\dot q}_i
 =
 -2(kq_{i,x}+c)\Big(p_i-\frac{1}{Nc}\sum\limits_{j=1}^Np_j(kq_{j,x}+c)\Big)\,,
 }
 \end{array}
 \eq
  \beq\label{q08}
  \begin{array}{c}
  \displaystyle{
 {\dot p}_i
 =
 -k\p_x\bigg( p_i^2-2p_i\frac{1}{Nc}\sum\limits_{j=1}^Np_j(kq_{j,x}+c)+
 \frac{1}{2}\frac{k^3 q_{i,xxx}}{kq_{i,x}+c}-\frac{1}{4}\frac{k^4 q^2_{i,xx}}{(kq_{i,x}+c)^2} \bigg)-
  }
  \\
  \displaystyle{
-2\sum\limits_{j:j\neq i}^N \bigg(
 (kq_{j,x}+c)^3\wp'(q_{i}-q_j)-3k^2(kq_{j,x}+c)q_{j,xx}\wp(q_{i}-q_j)-k^3q_{j,xxx}\zeta(q_{i}-q_j)
 \bigg)\,,
 }
 \end{array}
 \eq
 where the short notation for partial derivatives is used: ${\dot \varphi}=\p_t\varphi(t,x)$; $\varphi_x=\partial_x\varphi(t,x)$, $\varphi_{xx}=\partial^2_x\varphi(t,x)$ e.t.c., and $\zeta(x)$ is the Weierstrass $\zeta$-function.
 The parameter $k\in\mathbb C$ is a factor preceding $\p_x$, and its power measures the number of $x$-derivatives.
 The limit $k\rightarrow 0$ corresponds to the limit to finite-dimensional mechanics, when all the fields
  become independent of $x$. In this limit
 equations (\ref{q07})-(\ref{q08}) become the ordinary differential equations (\ref{q02}) with the
 normalization constant $\ti b=-2c$.

 In the field case the Lax equation (\ref{q01}) is replaced with the Zakharov-Shabat zero curvature equation \cite{ZaSh,FT}:
 \beq\label{q09}
 \begin{array}{c}
  \displaystyle{
 \partial_{t}{U}(z)-k\partial_{x}{V}(z)+[{U}(z), {V}(z)]=0\,,\qquad {U}(z), {V}(z)\in\Mat\,.
  }
 \end{array}
\eq
 The infinite-dimensional analogue of the Liouville integrability is as follows. The integrals of motion
 come from the generating function $\tr(T^n(z,2\pi))$, where
 $
 T(z,x)={\rm Pexp}\Big( \frac{1}{k}\int\limits_0^x dy\,  U(z,y) \Big)
$
is the monodromy matrix. The involution property $\{\tr(T^m(z,2\pi)),\tr(T^n(w,2\pi))\}=0$ follows
from existence of the classical $r$-matrix structure of Maillet type. For the spinless 1+1 Calogero-Moser model
(\ref{q07})-(\ref{q08}) the existence of such $r$-matrix structure was recently proved in \cite{Z24}.

{\bf The purpose of the paper} is to describe the field analogue of the spin Calogero-Moser model
(\ref{q03})-(\ref{q04}). In fact, a part of this problem was already solved. The Lax matrix (the $U$-matrix)
for the field analogue of the spin Calogero-Moser model was derived in \cite{LOZ} using
the group-theoretical approach to Hitchin systems on the affine Higgs bundles.
It was also shown that the
constraints $\mu_i={\rm const}$ (\ref{q05}) underlying additional reduction are extended as
$\mu_i(x)={\rm const}$ with
  \beq\label{q10}
  \begin{array}{c}
  \displaystyle{
 \mu_i(x)=\mS_{ii}(x)-k\p_xq_i(x)\,,\quad i=1,...,N\,,
 }
 \end{array}
 \eq
where $\mS_{ii}$ are diagonal elements of the spin field matrix $\mS=\mS(x)\in\Mat$.
These spin variables has the Poisson brackets
 $\{\mS_{ij}(x),\mS_{ml}(y)\}=(-\mS_{il}(x)\delta_{mj}+\mS_{mj}(x)\delta_{il})\delta(x-y)$
 on the unreduced phase space -- a coadjoint orbit of the loop group
  $L({\rm GL}_N)$ of ${\mathcal C}^\infty$-maps
${\mathbb S}^1\rightarrow {\rm GL}_N$. This allows to describe the classical $r$-matrix
structure with unwanted term on the unreduced phase space \cite{Z24}.

In this paper the equations of motion are deduced. For
this purpose we suggest an ansatz for $V$-matrix satisfying the Zakharov-Shabat equation
with (unwanted term) non-trivial right hand side:
 \beq\label{q11}
 \begin{array}{c}
  \displaystyle{
 \partial_{t}{U}(z)-k\partial_{x}{V}(z)+[{U}(z), {V}(z)]={\rm unwanted\ term}\,.
  }
 \end{array}
\eq
 Similarly to the finite-dimensional case (\ref{q05}) the unwanted term vanishes on-shell the constraints $\mu_i(x)={\rm const}$, thus providing integrable model on the reduced phase space.
 Our construction of $U$-$V$ pair works in the special case, when the matrix of spin field variables
 $\mS$ satisfies the property
 \beq\label{q12}
 \begin{array}{c}
  \displaystyle{
\mS^2=c_0\mS
  }
 \end{array}
\eq
 with some constant $c_0\in\mC$, that is eigenvalues of $\mS$ are equal to either $0$ or $c_0$.
 The equations of motion coming from (\ref{q09}) are as follows:
  \beq\label{q7410}
  \begin{array}{c}
  \displaystyle{
\p_t p_i+k\p_x { d}_i+\sum\limits_{m\neq i}^N\mS_{im}\mS_{mi}(\mS_{ii}+\mS_{mm})\wp'(q_{im})-
\sum\limits_{m\neq i}^N\Big( \mS_{im}{H}_{mi}-{H}_{im}\mS_{mi} \Big)\wp(q_{im})=0
  }
 \end{array}
 \eq
 and
  \beq\label{q7510}
  \begin{array}{c}
  \displaystyle{
\p_t \mS_{ij} -k\p_x{\breve H}_{ij}-2c_0(\mS_{ii}-\mS_{jj})\mS_{ij}\wp(q_{ij})+
  }
  \\ \ \\
   \displaystyle{
+(p_i-p_j)H_{ij}+\sum\limits_{m\neq i,j}^N\Big((\mS_{im}{H}_{mj}-{H}_{im}\mS_{mj})\eta(q_i,q_m,q_j)+
  }
  \\
   \displaystyle{
+\sum\limits_{m\neq i,j}^N\mS_{im}\mS_{mj}
\Big( (\mS_{ii}+\mS_{mm})\wp(q_{im})-(\mS_{jj}+\mS_{mm})\wp(q_{jm})+(\mS_{ii}-\mS_{jj})\wp(q_{ij}) \Big)=0
  }
 \end{array}
 \eq
 for $i\neq j$ and $i,j=1,...,N$, where the function $\eta(q_i,q_m,q_j)$ is given in (\ref{q35}) and
     %
  \beq\label{q7311}
  \begin{array}{c}
  \displaystyle{
 H=-\frac{k}{c_0}\,[\mS,\mS_x]\in\Mat\,,\qquad {\breve H}=H-O\in\Mat\,,
  }
 \end{array}
 \eq
  \beq\label{q661}
  \begin{array}{c}
  \displaystyle{
 {O}_{ij}=\delta_{ij}2(p_i+\kappa)\mS_{ii}+(1-\delta_{ij})\Big((p_i+p_j+2\kappa)\mS_{ij}+
  \sum\limits_{m\neq i,j}^N\mS_{im}\mS_{mj}\eta(q_i,q_m,q_j)\Big)\,,
  }
 \end{array}
 \eq
  \beq\label{q651}
  \begin{array}{c}
  \displaystyle{
  d_i=p_i^2+2\kappa p_i-\sum\limits_{m\neq i}^N\mS_{im}\mS_{mi}\wp(q_{im})\,,\quad i=1,...,N\,,
  \qquad \kappa=-\frac{1}{\tr(\mS)}\sum\limits_{i=1}^N\mS_{ii}p_i\,.
  }
 \end{array}
 \eq
 Two requirements are fulfilled for the obtained equations of motion. On the one hand
 equations generalize the finite-dimensional case, that is
 (\ref{q03})-(\ref{q04}) follow in the $k\rightarrow 0$ limit.
 On the other hand the spinless model (\ref{q07})-(\ref{q08})
  is reproduced in the special case ${\rm rank}(\mS)=1$.

Let us also mention that the field Calogero-Moser model is related to a certain class of Landau-Lifshitz models
through IRF-Vertex relations \cite{AtZ1}, where the condition $\mS^2=c_0\mS$ plays the same important role \cite{AtZ2}. Also, the spinless field Calogero-Moser model can be extended to the field analogue of the
Ruijsenaars-Schneider model described in terms of semi-discrete Zakharov-Shabat equation \cite{ZZ}.

{\bf The paper is organized as follows.} In the next Section we recall the description of the
 classical finite-dimensional spin Calogero-Moser model and reduction to the spinless case.
 In Section 3 we briefly review the $U$-$V$ pair of the spinless field Calogero-Moser model from \cite{Krich22}.
 In Section 4 we consider a natural ansatz for $V$-matrix and derive the corresponding equations.
 However, a relation to the spinless model cannot be established at this stage. This problem is solved
 in Section 5, and the final answer for $V$-matrix and equations of motion are obtained. Finally, it is shown in
 Section 6 how the obtained $U$-$V$ pair is reduced to the one for the spinless field Calogero-Moser model.
 Definitions and properties of elliptic functions are given in the Appendix.

\section{Spin Calogero-Moser model in classical mechanics}\label{sec2}
\setcounter{equation}{0}
\paragraph{Unreduced phase space.} The classical spin Calogero-Moser model of ${\rm gl}_N$ type is described by the following Lax pair:
  \beq\label{q21}
  \begin{array}{c}
  \displaystyle{
L^{\hbox{\tiny{spin}}}_{ij}(z)=\delta_{ij}(p_i+S_{ii}E_1(z))+(1-\delta_{ij})S_{ij}\phi(z,q_{i}-q_j)\,,
 }
 \end{array}
 \eq
  \beq\label{q22}
  \begin{array}{c}
  \displaystyle{
  M^{\hbox{\tiny{spin}}}_{ij}(z)=-(1-\delta_{ij}){\ti b}S_{ij}f(z,q_i-q_j)\,,
 }
 \end{array}
 \eq
where the functions $E_1(z)$, $\phi(z,u)$ and $f(z,u)$ are defined in the Appendix.
This Lax pair satisfies the Lax equation with the unwanted term (\ref{q05}). Namely,
  \beq\label{q23}
  \begin{array}{c}
  \displaystyle{
 \dot{L}^{\hbox{\tiny{spin}}}(z)+[L^{\hbox{\tiny{spin}}}(z),M^{\hbox{\tiny{spin}}}(z)]=-{\ti b}\sum\limits_{i,j=1}^N E_{ij}\,(S_{ii}-S_{jj})S_{ij}E_1(z)f(z,q_{ij})\,.
 }
 \end{array}
 \eq
The latter equation provides the equations of motion (\ref{q03})-(\ref{q04}).
The r.h.s. of (\ref{q23}) vanishes on the constraints $S_{ii}={\rm const}$, $i=1,...,N$.
These are the first class constraints generated by the coadjoint action of the Cartan subgroup $\mH\subset {\rm GL}_N$
on the coadjoint orbit $\mO$ described by the spin matrix $S$ (with fixed eigenvalues).
Another $N$ constraints are some gauge fixation conditions -- fixation of the action of $\mH$.
Totally we have $2N$ conditions, and these are the second class constraints. One can perform
the Poisson reduction using the Dirac formula. The Dirac terms may change the Poisson brackets $\{S_{ij},S_{kl}\}$
and equations of motion. Notice that the final answer
(for equations of motion) depends on the choice of the gauge fixation
conditions, and there is no preferred way to choose them. This is why the classical spin
Calogero-Moser model is usually described at the level of the unreduced phase space.

\paragraph{Spinless model.}  Reduction to the spinless model (\ref{q02}) is performed as follows \cite{BAB}.
 Consider a special case when the matrix of spin
variables is of rank one:
\beq\label{q24}
  \begin{array}{c}
  \displaystyle{
 S_{ij}=\xi_i\psi_j\,.
 }
 \end{array}
 \eq
 Next, perform the gauge transformation
 %
  \beq\label{q25}
  \begin{array}{c}
  \displaystyle{
L^{\hbox{\tiny{spin}}}(z)\rightarrow g^{-1}L^{\hbox{\tiny{spin}}}(z)g\,,\qquad
M^{\hbox{\tiny{spin}}}(z)\rightarrow g^{-1}M^{\hbox{\tiny{spin}}}(z)g-g^{-1}{\dot g}
 }
 \end{array}
 \eq
 with
\beq\label{q26}
  \begin{array}{c}
  \displaystyle{
 g={\rm diag}(\xi_1,...,\xi_N)\in\Mat\,.
 }
 \end{array}
 \eq
 The action on the matrix $S$ is
\beq\label{q27}
  \begin{array}{c}
  \displaystyle{
 S_{ij}\rightarrow \xi_j\psi_j=S_{jj}\,,
 }
 \end{array}
 \eq
 and $S_{ii}={\rm const}$ due to (\ref{q04}), that is the spin degrees of freedom
 are gauged away. Indeed, since the Lax matrix defines all Hamiltonian flows (and equations of motion)
 we conclude that the model is independent of spin variables in this case.

 Due to (\ref{q25}) the $M$-matrix acquires additional diagonal part
\beq\label{q28}
  \begin{array}{c}
  \displaystyle{
-\Big(g^{-1}{\dot g}\Big)_{ij}=-\delta_{ij}\frac{\dot \xi_i}{\xi_i}\,.
 }
 \end{array}
 \eq
It can be computed through the Hamiltonian ${\ti b}H^{\hbox{\tiny{spin}}}$, where
  \beq\label{q29}
  \begin{array}{c}
  \displaystyle{
H^{\hbox{\tiny{spin}}}=\sum\limits_{i=1}^N\frac{p_i^2}{2}-\sum\limits_{i>j}^N S_{ij}S_{ji}\wp(q_i-q_j)
 }
 \end{array}
 \eq
 and the Poisson brackets $\{\xi_i,\eta_j\}=\delta_{ij}$. In this way the spinless model arise\footnote{Some more comments on the reduction to spinless model can be found in \cite{BAB}, see also \cite{Z24}.}.
 The Lax pair of the spinless model \cite{Kr}
  \beq\label{q291}
  \begin{array}{c}
  \displaystyle{
L^{\hbox{\tiny{CM}}}_{ij}(z)=\delta_{ij}(p_i+{c} E_1(z))+{c}(1-\delta_{ij})\phi(z,q_{ij})\,,\quad q_{ij}=q_i-q_j\,,
 }
 \end{array}
 \eq
  \beq\label{q292}
  \begin{array}{c}
  \displaystyle{
M^{\hbox{\tiny{CM}}}_{ij}(z)=-{\ti b}{c}\delta_{ij} d_i
-{\ti b}{c}(1-\delta_{ij})f(z,q_{ij})\,,\quad d_i=\sum\limits_{k:k\neq i}^N
\wp(q_{ik})
 }
 \end{array}
 \eq
satisfies the Lax equation (\ref{q01}) and provides equations (\ref{q02}).

Alternatively, in order to deduce the spinless model one can use the equations of motion (\ref{q04}).
In the case of ${\rm rank}(S)=1$ (\ref{q24}) we have
  \beq\label{q293}
  \begin{array}{c}
  \displaystyle{
{\dot S}_{ij}=\Big(\frac{\dot \xi_i}{\xi_i}+\frac{\dot \psi_j}{\psi_j}\Big)\xi_i\psi_j\,.
 }
 \end{array}
 \eq
On the other hand from (\ref{q04}) we have
  \beq\label{q294}
  \begin{array}{c}
  \displaystyle{
 {\dot S}_{ij}={\ti b}\xi_i\psi_j\sum\limits_{m\neq i,j}^N
\xi_m\psi_m(\wp(q_i-q_m)-\wp(q_j-q_m))\,,\ i\neq j\,.
 }
 \end{array}
 \eq
The reduction consists of constraints (\ref{q06}) $S_{ii}=c$ and some gauge fixing conditions, which
we choose for definiteness to be $\xi_i=\psi_i=c^{1/2}$, i.e. on-shell the constraints we have $S_{ij}=c$,
and from (\ref{q294}) one gets
  \beq\label{q295}
  \begin{array}{c}
  \displaystyle{
 {\dot S}_{ij}\Big|_{\rm on-shell}={\ti b}c^2\sum\limits_{m\neq i,j}^N
(\wp(q_i-q_m)-\wp(q_j-q_m))\,,\ i\neq j\,.
 }
 \end{array}
 \eq
Therefore, from (\ref{q293}) we conclude that
  \beq\label{q296}
  \begin{array}{c}
  \displaystyle{
 \frac{\dot \xi_i}{\xi_i}\Big|_{\rm on-shell}={\ti b}c^2\sum\limits_{m\neq i}^N
\wp(q_i-q_m)+{\dot \lambda}\,,
 }
 \end{array}
 \eq
where $\lambda(t)$ is some function coming from a freedom in definition of $\xi_i$, $\psi_j$:
$\xi_i\rightarrow e^\lambda\xi_i$, $\psi_j\rightarrow e^{-\lambda}\psi_j$.
Finally, plugging (\ref{q296}), (\ref{q28}) into (\ref{q25}) yields
  \beq\label{q297}
  \begin{array}{c}
  \displaystyle{
 \Big(g^{-1}M^{\hbox{\tiny{spin}}}(z)g-g^{-1}{\dot g}\Big)\Big|_{\rm on-shell}=M^{\hbox{\tiny{CM}}}(z)-1_N{\dot \lambda}
 }
 \end{array}
 \eq
The term $1_N{\dot \lambda}$ does not provide any input into the Lax equation, and it can be gauged away.

\section{Spinless field Calogero-Moser model}\label{sec3}
\setcounter{equation}{0}
 Here we recall main result of \cite{Krich22}. Below we use slightly different notations (different from \cite{Krich22}), normalization constants and
 gauge choice for $U$-$V$ pair, see \cite{Z24}. From now on $q_i=q_i(t,x)$, $p_i=p_i(t,x)$, $i=1,...,N$ and $\mS_{ij}=\mS_{ij}(x)$, $i,j=1,...,N$ are periodic fields on a circle, and the time dependence comes from equations
 of motion generated by the Zakharov-Shabat equation.

 Introduce notations
  \beq\label{q301}
  \begin{array}{c}
  \displaystyle{
 \al_i=(kq_{i,x}+c)^{1/2}\,, \quad i=1,...,N
 }
 \end{array}
 \eq
and
  \beq\label{q302}
  \begin{array}{c}
  \displaystyle{
  \kappa=-\frac{1}{Nc}\sum\limits_{j=1}^Np_j(kq_{j,x}+c)=-\frac{1}{Nc}\sum\limits_{j=1}^Np_j\al_j^2\,.
 }
 \end{array}
 \eq
In what follows we also use the short notation
  \beq\label{q30}
  \begin{array}{c}
  \displaystyle{
  q_{ij}=q_i-q_j
 }
 \end{array}
 \eq
as in (\ref{q291})-(\ref{q292}).
Consider the $U$-$V$ pair
  \beq\label{q31}
  \begin{array}{c}
  \displaystyle{
 U^{\hbox{\tiny{2dCM}}}_{ij}(z)=\delta_{ij}
 \Big( p_i+\al_i^2E_1(z)-k\frac{\al_{i,x}}{\al_i} \Big)
 +(1-\delta_{ij})\phi(z,q_i-q_j)\al_j^2\,,
  }
 \end{array}
 \eq
  \beq\label{q32}
  \begin{array}{c}
  \displaystyle{
 V^{\hbox{\tiny{2dCM}}}_{ij}(z)=\delta_{ij}
 \Big( q_{i,t}E_1(z)-Nc\al_i^2\wp(z)-m_i^0-\frac{\al_{i,t}}{\al_i} \Big)+
   }
  \\ \ \\
  \displaystyle{
 +(1-\delta_{ij})\Big(Ncf(z,q_i-q_j)-NcE_1(z)\phi(z,q_i-q_j)-m_{ij}\phi(z,q_i-q_j)\Big)\al_j^2\,,
  }
 \end{array}
 \eq
 where
  \beq\label{q33}
  \begin{array}{c}
  \displaystyle{
 m_i^0=p_i^2+2\kappa p_i+k^2\frac{\al_{i,xx}}{\al_i}-
 \sum\limits_{j:j\neq i}^N\Big(
 (2\al_j^4+\al_i^2\al_j^2)\wp(q_i-q_j)+4k\al_j\al_{j,x}\zeta(q_i-q_j)
 \Big)\,,
  }
 \end{array}
 \eq
 and
  \beq\label{q34}
  \begin{array}{c}
  \displaystyle{
 m_{ij}=p_i+p_j+2\kappa-k\frac{\al_{i,x}}{\al_i}+k\frac{\al_{j,x}}{\al_j}+
 \sum\limits_{l:l\neq i,j}^N \al_l^2\eta(q_i,q_l,q_j)\,,\quad i\neq j
  }
 \end{array}
 \eq
 with
  \beq\label{q35}
  \begin{array}{c}
  \displaystyle{
 \eta(z_1,z_2,z_3)=\zeta(z_1-z_2)+\zeta(z_2-z_3)+\zeta(z_3-z_1)\stackrel{(\ref{a05})}{=}
  }
  \\ \ \\
    \displaystyle{
=
 E_1(z_1-z_2)+E_1(z_2-z_3)+E_1(z_3-z_1)\,.
  }
 \end{array}
 \eq
 Then the Zakharov-Shabat equation (\ref{q09}) for (\ref{q31})-(\ref{q32}) provides
 the equations of motion (\ref{q07})-(\ref{q08}), which are rewritten in the notations (\ref{q301})-(\ref{q302})
 as
  \beq\label{q36}
  \begin{array}{c}
  \displaystyle{
 {\dot q}_i=-2\al_i^2(p_i+\kappa)\,,
 }
 \end{array}
 \eq
  \beq\label{q37}
  \begin{array}{c}
  \displaystyle{
 {\dot p}_i=-k\p_x\Big( p_i^2+2\kappa p_i+k^2\frac{\al_{i,xx}}{\al_i} \Big)
 -2\sum\limits_{j:j\neq i}^N
 \Big( \al_j^6\wp'(q_{ij})-6\al_j^3\al_{j,x}\wp(q_{ij})-k^2\p_x^2(\al_j^2)\zeta(q_{ij}) \Big)\,.
  }
 \end{array}
 \eq
 These equations are Hamiltonian with the Hamiltonian function
 %
 \beq\label{q38}
 \begin{array}{c}
  \displaystyle{
\mH^{\hbox{\tiny{2dCM}}}=\int\limits_{0}^{2\pi} d x H^{\hbox{\tiny{2dCM}}}(x)\,,\qquad
H^{\hbox{\tiny{2dCM}}}(x)=-\sum\limits_{i=1}^N p_i^2\al_i^2+Nc\kappa^2
+\sum\limits_{i=1}^N k^2\al_{i,x}^2+
  }
  \\
  \displaystyle{
+\frac{k}{2}\sum\limits_{i\neq j}^N\Big( \al_i\al_{j,x}-\al_j\al_{i,x}+c(\al_{i,x}-\al_{j,x}) \Big)\zeta(q_{ij})
+\frac{1}{2}\sum\limits_{i\neq j}^N\Big( \al_i^4\al_j^2+\al_i^2\al_j^4-c(\al_i^2-\al_j^2)^2 \Big)\wp(q_{ij})\,.
 }
 \end{array}
 \eq
In the limit to the finite-dimensional case $k=0$. Then $\kappa=\frac{1}{N}\sum\limits_{m=1}^Np_m$, $\al_i^2=c$,
and we come to (\ref{q02}) with ${\ti b}=-2c$.

\section{1+1 field spin Calogero-Moser model: ansatz for $U$-$V$ pair}\label{sec4}
\setcounter{equation}{0}
The $U$-matrix for the 1+1 field spin Calogero-Moser model was derived in \cite{LOZ} by
extending the corresponding moment map equation (defining relation for the Lax matrix)
to the case of loop group. The result is
  \beq\label{q40}
  \begin{array}{c}
  \displaystyle{
 U^{\hbox{\tiny{2dSpin}}}_{ij}(z)=
 \delta_{ij}(p_i+\mS_{ii}(x)E_1(z))+(1-\delta_{ij})\mS_{ij}(x)\phi(z,q_i(x)-q_j(x))\,,
  }
 \end{array}
 \eq
 and it has exactly the same form as in the finite dimensional case (\ref{q21}) but the additional
 constraint is different. It is now
  \beq\label{q41}
  \begin{array}{c}
  \displaystyle{
\mu_i(x)=c\,,\quad i=1,...,N
 }
 \end{array}
 \eq
 where $\mu_i(x)$ is (\ref{q10}) and $c\in\mC$ is some constant. Equivalently,
  \beq\label{q42}
  \begin{array}{c}
  \displaystyle{
\mS_{ii}(x)-k\p_xq_{i}(x)=c\quad\hbox{or}\quad \mS_{ii}(x)=\al_i^2(x)\,,\quad i=1,...,N\,.
 }
 \end{array}
 \eq
However, the group-theoretical approach of \cite{LOZ} does not provide any constructive approach for
explicit calculation of $V$-matrix. Here we use an ansatz, which is similar to the one
for Landau-Lifshitz equation \cite{Skl} and its ${\rm gl}_N$ generalization \cite{Z11,AtZ2}. While matrix $U$
has the first order pole at $z=0$, the matrix $V$ has the second and the first order poles
at $z=0$. The coefficient before the second order pole should be proportional to $\mS=\res\limits_{z=0}U(z)$ in order not to produce the third order pole from the commutator $[U(z),V(z)]$. The residue $\res\limits_{z=0}V(z)$
at the first order pole is some arbitrary matrix ${\ti H}\in\Mat$ to be defined from the Zakharov-Shabat equations.
Consider  the following $V$-matrix:
  \beq\label{q43}
  \begin{array}{c}
  \displaystyle{
  V_{ij}(z)=
 \delta_{ij}\Big({\ti d}_i+{\ti H}_{ii}E_1(z)+b\mS_{ii}\wp(z)\Big)+
 (1-\delta_{ij})\Big({\ti H}_{ij}\phi(z,q_{ij})+b\mS_{ij}F(z,q_{ij})\Big)\,,
  }
 \end{array}
 \eq
 where the function $F(z,q)$ is given by (\ref{a12}) and ${\ti H}={\ti H}(t,x)\in\Mat$ is some matrix
 to be defined.

\begin{theor}
The $U$-$V$ pair (\ref{q40}) and (\ref{q43}) satisfies the Zakharov-Shabat equation with the unwanted term
  \beq\label{q44}
  \begin{array}{c}
  \displaystyle{
\partial_{t}{U}^{\hbox{\tiny{2dSpin}}}(z)-k\partial_{x}{V}(z)+[{U}^{\hbox{\tiny{2dSpin}}}(z), {V}(z)]=
  }
  \\
  \displaystyle{
=\sum\limits_{i\neq j}^N E_{ij}\Big(b\mS_{ij}F'(z,q_{ij})+{\ti H}_{ij} f(z,q_{ij})\Big)
\Big(\mu_i(x)-\mu_j(x)\Big)\,,
  }
  \\
  \mu_i(x)=\mS_{ii}-k\p_xq_{i}
 \end{array}
 \eq
%
and provides the following set of equations:
  \beq\label{q45}
  \begin{array}{c}
  \displaystyle{
\p_t p_i-k\p_x {\ti d}_i+b\sum\limits_{m\neq i}^N\mS_{im}\mS_{mi}\wp'(q_{im})-
\sum\limits_{m\neq i}^N\Big( \mS_{im}{\ti H}_{mi}-{\ti H}_{im}\mS_{mi} \Big)\wp(q_{im})=0\,,
  }
 \end{array}
 \eq
  \beq\label{q46}
  \begin{array}{c}
  \displaystyle{
\p_t \mS_{ij} -k\p_x{\ti H}_{ij}+2b(\mS_{ii}-\mS_{jj})\mS_{ij}\wp(q_{ij})
+(p_i-p_j){\ti H}_{ij}-({\ti d}_i-{\ti d}_j)\mS_{ij}+
  }
  \\
   \displaystyle{
+\sum\limits_{m\neq i,j}^N\Big((\mS_{im}{\ti H}_{mj}-{\ti H}_{im}\mS_{mj})\eta(q_i,q_m,q_j)
 -b\mS_{im}\mS_{mj}(\wp(q_{im})-\wp(q_{mj}))\Big)=0
\quad\hbox{for}\ i\neq j\,,
  }
 \end{array}
 \eq
  \beq\label{q47}
  \begin{array}{c}
  \displaystyle{
\p_t \mS_{ii}-k\p_x{\ti H}_{ii}=0\,,
  }
 \end{array}
 \eq
  \beq\label{q48}
  \begin{array}{c}
  \displaystyle{
 \p_t q_i={\ti H}_{ii}+g(t,x)
  }
 \end{array}
 \eq
 and
  \beq\label{q49}
  \begin{array}{c}
  \displaystyle{
bk\p_x \mS=[\mS,{\ti H}-bP]\,,\quad P={\rm diag}(p_1,...,p_N)\,.
  }
 \end{array}
 \eq
 In (\ref{q46}) the function $\eta(q_i,q_m,q_j)$ from (\ref{q35}) is used.
\end{theor}
\begin{proof}
The proof is by straightforward calculation using identities from the Appendix.
The diagonal part of (\ref{q44}) provides (\ref{q47}) as coefficient
before $E_1(z)$, the diagonal part of (\ref{q49}) as coefficient
before $\wp(z)$ and the equation (\ref{q45}) as a constant (in $z$) term. Expressions
in the off-diagonal part of (\ref{q44}) are expanded in linear combination of
the following (linearly independent) functions: $\phi(z,q_{ij})$, $f(z,q_{ij})$,
$E_1(z)\phi(z,q_{ij})$, $E_1(z)f(z,q_{ij})$ and $(E_1^2(z)-\wp(z))\phi(z,q_{ij})$.
The coefficients before $(E_1^2(z)-\wp(z))\phi(z,q_{ij})$ and $E_1(z)f(z,q_{ij})$ are identities
 due to definition of $\mu_i(x)$, the coefficient before $E_1(z)\phi(z,q_{ij})$ provides the off-diagonal part of
 (\ref{q49}), the coefficient before $\phi(z,q_{ij})$ provides (\ref{q46}) and
 the coefficient before $f(z,q_{ij})$ provides (\ref{q48}).
\end{proof}

In what follows we fix $g(t,x)=0$ in (\ref{q48}). Notice also that the condition (\ref{q47}) with ${\ti H}_{ii}={\dot q}_i$
is in agreement with the constraints $\mu_i(x)=c$ since (\ref{q47}) is rewritten as $\p_t\mu_i(x)=0$.

The equations of motion (\ref{q45}), (\ref{q46}), (\ref{q48}) are written in terms of unknown matrix ${\ti H}$,
which should be defined from (\ref{q49}). Consider $H\in\Mat$ satisfying
  \beq\label{q50}
  \begin{array}{c}
  \displaystyle{
bk\p_x \mS=[\mS, H]\,.
  }
 \end{array}
 \eq
Then ${\ti H}=H+bP$. Equation (\ref{q50}) appears also in the description of higher rank Landau-Lifshitz and 1+1 Gaudin type models (see e.g. \cite{AtZ2,Z11}). A possible way to solve it is to impose condition
  \beq\label{q51}
  \begin{array}{c}
  \displaystyle{
{\mathcal S}^2=c_0{\mathcal S}\,,
  }
 \end{array}
 \eq
 where $c_0\in\mC$ is some constant. Relation (\ref{q51}) means that the eigenvalues of matrix $S$ are equal
 to either $0$ or $c_0$. From a general construction presented in \cite{LOZ} it follows that eigenvalues
 of $\mS$ are conserved quantities (in fact, they are the Casimir functions). Therefore,
 the condition (\ref{q51}) is preserved on the equations of motion.

 From (\ref{q51}) it follows that $\mS_x\mS+\mS\mS_x=c_0\mS_x$ (where $\mS_x=\p_x\mS$). Then it is easy to verify that (\ref{q50}) has solution
  \beq\label{q52}
  \begin{array}{c}
  \displaystyle{
 H=\frac{bk}{c_0^2}\,[\mS,\mS_x]\,,
  }
 \end{array}
 \eq
 that is
  \beq\label{q53}
  \begin{array}{c}
  \displaystyle{
 {\ti H}=bP+\frac{bk}{c_0^2}\,[\mS,\mS_x]\,.
  }
 \end{array}
 \eq
 Its diagonal part provides via (\ref{q48}) (with $g(t,x)=0$)
  \beq\label{q54}
  \begin{array}{c}
  \displaystyle{
 {\dot q}_i=bp_i+\frac{bk}{c_0^2}\,[\mS,\mS_x]_{ii}\,,
  }
 \end{array}
 \eq
 and the off-diagonal part of $\ti H$ enters equations (\ref{q45})-(\ref{q46}).

 The diagonal elements ${\ti d}_i$ remain undefined. At the same time we of course need to
 reproduce the finite-dimensional equations (\ref{q03})-(\ref{q04}) when all the fields are independent
 of $x$ (equivalently, when $k=0$). When $k=0$ we have $H=0$, ${\ti H}=bP$ and
 $U^{\hbox{\tiny{spin}}}(z)=L^{\hbox{\tiny{spin}}}(z)$. If ${\ti d}_i(k=0)=0$ then
 $V(z)=E_1(z)L^{\hbox{\tiny{spin}}}(z)+M^{\hbox{\tiny{spin}}}(z)$ with $\ti b=b$.
 Obviously, redefinition $M^{\hbox{\tiny{spin}}}(z)\rightarrow E_1(z)L^{\hbox{\tiny{spin}}}(z)+M^{\hbox{\tiny{spin}}}(z)$ does not effect the Lax equation (\ref{q05}).
 In this way we come to (\ref{q03})-(\ref{q04}).

\section{Specifying the ansatz for $V$-matrix}\label{sec5}
\setcounter{equation}{0}

There is a problem with the above given description. It is not in agreement with the spinless case (\ref{q31})-(\ref{q37}). As will be explained below, the equation for ${\dot q}_i$ (\ref{q54}) can not be transformed to (\ref{q36})
through the reduction to the spinless model. Another disagreement is that for $V(z)$ (\ref{q43})
  \beq\label{q55}
  \begin{array}{c}
  \displaystyle{
 V(z)|_{k=0}=E_1(z)L^{\hbox{\tiny{spin}}}(z)+M^{\hbox{\tiny{spin}}}(z)|_{\ti b=b}\,,
  }
 \end{array}
 \eq
 while
  \beq\label{q56}
  \begin{array}{c}
  \displaystyle{
 V^{\hbox{\tiny{2dCM}}}(z)|_{k=0}=\Big(L^{\hbox{\tiny{CM}}}(z)\Big)^2
 -\frac{2}{N}\Big(\sum\limits_{m=0}^Np_m\Big)L^{\hbox{\tiny{CM}}}(z)+M^{\hbox{\tiny{CM}}}(z)|_{\ti b=-2c}
  }
 \end{array}
 \eq
up to constant matrices proportional to the identity matrix $1_N$.

\subsection{Another role of condition ${\mathcal S}^2=c_0{\mathcal S}$}
Besides solving (\ref{q50}) we are going to use the condition (\ref{q51}) to modify
the definition of $V$-matrix.
\begin{predl}
Suppose the condition (\ref{q51}) holds true. Then the matrix $V^0(z)\in\Mat$
  \beq\label{q57}
  \begin{array}{c}
  \displaystyle{
 V^0_{ij}(z)=-\Big[\Big(U^{\hbox{\tiny{2dSpin}}}(z)\Big)^2\Big]_{ij}
+
  }
  \\ \ \\
   \displaystyle{
 +\delta_{ij}(\mS_{ii})^2\Big(E_1^2(z)-\wp(z)\Big)
 +(1-\delta_{ij})(\mS_{ii}+\mS_{jj})\mS_{ij}f(z,q_{ij})
  }
 \end{array}
 \eq
 has the form of the matrix $V(z)$ (\ref{q43})
 with
  \beq\label{q58}
  \begin{array}{c}
  \displaystyle{
 b=-c_0\,.
  }
 \end{array}
 \eq
\end{predl}
\begin{proof}
The following statement is verified directly:
  \beq\label{q59}
  \begin{array}{c}
  \displaystyle{
\Big[\Big(U^{\hbox{\tiny{2dSpin}}}(z)\Big)^2\Big]_{ij}=
  }
  \\ \ \\
   \displaystyle{
=\delta_{ij}\Big( p_i^2-\sum\limits_{m\neq i}^N\mS_{im}\mS_{mi}\wp(q_{im})+{\ti O}_{ii}E_1(z)
  +(\mS_{ii})^2\Big(E_1^2(z)-\wp(z)\Big)+(\mS^2)_{ii}\wp(z)\Big)+
  }
  \\ \ \\
   \displaystyle{
+(1-\delta_{ij})\Big((\mS^2)_{ij}F(z,q_{ij})+{\ti O}_{ij}\phi(z,q_{ij})+(\mS_{ii}+\mS_{jj})\mS_{ij}f(z,q_{ij})\Big)\,,
  }
 \end{array}
 \eq
 where the matrix ${\ti O}\in\Mat$ is as follows:
  \beq\label{q60}
  \begin{array}{c}
  \displaystyle{
 {\ti O}_{ij}=\delta_{ij}2p_i\mS_{ii}+(1-\delta_{ij})\Big((p_i+p_j)\mS_{ij}+
  \sum\limits_{m\neq i,j}^N\mS_{im}\mS_{mj}\eta(q_i,q_m,q_j)\Big)\,,
  }
 \end{array}
 \eq
where $\eta(q_i,q_m,q_j)$ is the function (\ref{q35}).
 Plugging ${\mathcal S}^2=c_0{\mathcal S}$ we get:
  \beq\label{q591}
  \begin{array}{c}
  \displaystyle{
V^0_{ij}(z)
=-\delta_{ij}\Big( p_i^2-\sum\limits_{m\neq i}^N\mS_{im}\mS_{mi}\wp(q_{im})+{\ti O}_{ii}E_1(z)
  +c_0\mS_{ii}\wp(z)\Big)-
  }
  \\ \ \\
   \displaystyle{
-(1-\delta_{ij})\Big(c_0\mS_{ij}F(z,q_{ij})
+{\ti O}_{ij}\phi(z,q_{ij})\Big)\,.
  }
 \end{array}
 \eq
 This finishes the proof.
\end{proof}

The defined above matrix ${\ti O}$ (\ref{q60}) has an important property.
\begin{predl}
Suppose the condition (\ref{q51}) holds true. Then
  \beq\label{q61}
  \begin{array}{c}
  \displaystyle{
  [\mS,{\ti O}-c_0P]=0\,,
  }
 \end{array}
 \eq
 where $P={\rm diag}(p_1,...,p_N)$.
\end{predl}
\begin{proof}
The proof is based on skew-symmetry of the function $\eta(q_i,q_m,q_j)$ (\ref{q35}) with respect to permutation of
any pair of arguments. The latter follows from (\ref{a05})-(\ref{a06}). Also,
it is convenient to use a simple relation
  \beq\label{q62}
  \begin{array}{c}
  \displaystyle{
  \eta(q_m,q_l,q_j)-\eta(q_i,q_m,q_l)=\eta(q_i,q_l,q_j)-\eta(q_i,q_m,q_j)\,,
  }
 \end{array}
 \eq
 which follows from the definition (\ref{q35}).
\end{proof}

In fact, the relation (\ref{q61}) can be proved indirectly. We have two different representations (\ref{q57})
and (\ref{q591}) for the matrix $V^0(z)$. One can use both representations to compute the commutator $[U^{\hbox{\tiny{2dSpin}}},V^0(z)]$. For computation of the commutator in representation (\ref{q591})
one should use identities (\ref{a07})-(\ref{a11}), while for the representation (\ref{q57})
formula (\ref{a14}) should be used. By comparing both calculations one gets two different representations for
certain expressions. In this way (\ref{q61}) follows as well as two more nontrivial relations.
\begin{predl}
Suppose the condition (\ref{q51}) is valid.  Then the following relations hold true:
  \beq\label{q611}
  \begin{array}{c}
  \displaystyle{
  c_0\sum\limits_{m\neq i}^N\mS_{im}\mS_{mi}\wp'(q_{im})-
  \sum\limits_{m\neq i}^N\Big(\mS_{im}{\ti O}_{mi}-{\ti O}_{im}\mS_{mi}\Big)\wp(q_{im})
  =-\sum\limits_{m\neq i}^N\mS_{im}\mS_{mi}(\mS_{ii}+\mS_{mm})\wp'(q_{im})
  }
 \end{array}
 \eq
for any $i=1,...,N$ and
  \beq\label{q612}
  \begin{array}{c}
  \displaystyle{
  (p_i-p_j){\ti O}_{ij}-(d_i-d_j)\mS_{ij}+
  \sum\limits_{m\neq i,j}^N\Big( \mS_{im}{\ti O}_{mj}-{\ti O}_{im}\mS_{mj} \Big)\eta(q_i,q_m,q_j)
    }
  \\
   \displaystyle{
   -c_0\sum\limits_{m\neq i,j}^N \mS_{im}\mS_{mj}(\wp(q_{im})-\wp(q_{mj}))
   = (\mS_{ii}-\mS_{jj})\wp(q_{ij})\sum\limits_{m\neq i,j}^N\mS_{im}\mS_{mj}+
      }
  \\
   \displaystyle{
  + \sum\limits_{m\neq i,j}^N\mS_{im}\mS_{mj}
\Big( (\mS_{ii}+\mS_{mm})\wp(q_{im})-(\mS_{jj}+\mS_{mm})\wp(q_{jm}) \Big)
  \quad\hbox{for}\ i\neq j\,.
  }
 \end{array}
 \eq
\end{predl}
We will use these relations below to transform equations of motion.

\subsection{The final answer for $V$-matrix and equations of motion}
Introduce the matrix  $V^{\hbox{\tiny{2dSpin}}}\in\Mat$ with elements
  \beq\label{q63}
  \begin{array}{c}
  \displaystyle{
 V^{\hbox{\tiny{2dSpin}}}_{ij}(z)=V^0_{ij}(z)
 -2\kappa U^{\hbox{\tiny{2dSpin}}}_{ij}(z)+
\delta_{ij}H_{ii}E_1(z)
 +(1-\delta_{ij})H_{ij}\phi(z,q_{ij})
  }
 \end{array}
 \eq
or, equivalently,
  \beq\label{q64}
  \begin{array}{c}
  \displaystyle{
 V^{\hbox{\tiny{2dSpin}}}(z)=-\delta_{ij}\Big( d_{i}+{ O}_{ii}E_1(z)
  +c_0\mS_{ii}\wp(z)\Big)-(1-\delta_{ij})\Big(c_0\mS_{ij}F(z,q_{ij})+{O}_{ij}\phi(z,q_{ij})\Big)+
  }
  \\ \ \\
   \displaystyle{
    +\delta_{ij}H_{ii}E_1(z)
 +(1-\delta_{ij})H_{ij}\phi(z,q_{ij})  \,,
  }
 \end{array}
 \eq
where
  \beq\label{q65}
  \begin{array}{c}
  \displaystyle{
  d_i=p_i^2+2\kappa p_i-\sum\limits_{m\neq i}^N\mS_{im}\mS_{mi}\wp(q_{im})\,,\quad i=1,...,N
  }
 \end{array}
 \eq
and the matrix $O$ is defined as
  \beq\label{q66}
  \begin{array}{c}
  \displaystyle{
 {O}_{ij}=\delta_{ij}2(p_i+\kappa)\mS_{ii}+(1-\delta_{ij})\Big((p_i+p_j+2\kappa)\mS_{ij}+
  \sum\limits_{m\neq i,j}^N\mS_{im}\mS_{mj}\eta(q_i,q_m,q_j)\Big)\,.
  }
 \end{array}
 \eq
 The variable $\kappa$ is free, it will be fixed below.

 Notice that similarly to $\ti O$ the matrix $O$  satisfies the property
  \beq\label{q67}
  \begin{array}{c}
  \displaystyle{
  [\mS,{ O}-c_0P]=0
  }
 \end{array}
 \eq
 due to  $O$ is obtained from $\ti O$ by the shift $p_i\rightarrow p_i+\kappa$:
 $O={\ti O}+2\kappa\mS$, and $\ti O$ satisfies (\ref{q61}).

Let us formulate the main statement.
\begin{theor}
Suppose the condition (\ref{q51}) holds true. The $U$-$V$ pair (\ref{q40}) and (\ref{q64}) satisfies the Zakharov-Shabat equation with the unwanted term
  \beq\label{q68}
  \begin{array}{c}
  \displaystyle{
\partial_{t}{U}^{\hbox{\tiny{2dSpin}}}(z)-k\partial_{x}{V}^{\hbox{\tiny{2dSpin}}}(z)
+[{U}^{\hbox{\tiny{2dSpin}}}(z), {V}^{\hbox{\tiny{2dSpin}}}(z)]=
  }
  \\ \ \\
  \displaystyle{
=\sum\limits_{i\neq j}^N E_{ij}\Big(-c_0\mS_{ij}F'(z,q_{ij})+{\breve H}_{ij} f(z,q_{ij})\Big)
\Big(\mu_i(x)-\mu_j(x)\Big)\,,
  }
 \end{array}
 \eq
where
  \beq\label{q681}
  \begin{array}{c}
  \displaystyle{
  \mu_i(x)=\mS_{ii}-k\p_xq_{i}\,,
  }
  \\ \ \\
  {\breve H}=H-O\in\Mat
 \end{array}
 \eq
%
and provides the following set of equations:
  \beq\label{q69}
  \begin{array}{c}
  \displaystyle{
\p_t p_i+k\p_x { d}_i-c_0\sum\limits_{m\neq i}^N\mS_{im}\mS_{mi}\wp'(q_{im})-
\sum\limits_{m\neq i}^N\Big( \mS_{im}{\breve H}_{mi}-{\breve H}_{im}\mS_{mi} \Big)\wp(q_{im})=0\,,
  }
 \end{array}
 \eq
  \beq\label{q70}
  \begin{array}{c}
  \displaystyle{
\p_t \mS_{ij} -k\p_x{\breve H}_{ij}-2c_0(\mS_{ii}-\mS_{jj})\mS_{ij}\wp(q_{ij})
+(p_i-p_j){\breve H}_{ij}+({d}_i-{d}_j)\mS_{ij}+
  }
  \\
   \displaystyle{
+\sum\limits_{m\neq i,j}^N\Big((\mS_{im}{\breve H}_{mj}-{\breve H}_{im}\mS_{mj})\eta(q_i,q_m,q_j)
 +c_0\mS_{im}\mS_{mj}(\wp(q_{im})-\wp(q_{mj}))\Big)=0
\quad\hbox{for}\ i\neq j\,,
  }
 \end{array}
 \eq
  \beq\label{q71}
  \begin{array}{c}
  \displaystyle{
\p_t \mS_{ii}-k\p_x{\breve H}_{ii}=0\,,
  }
 \end{array}
 \eq
  \beq\label{q72}
  \begin{array}{c}
  \displaystyle{
 \p_t q_i={\breve H}_{ii}
  }
 \end{array}
 \eq
 and
  \beq\label{q73}
  \begin{array}{c}
  \displaystyle{
-c_0k\p_x \mS=[\mS,{\breve H}+c_0P]\,,\quad P={\rm diag}(p_1,...,p_N)\,.
  }
 \end{array}
 \eq
 The matrix $O$ entering the definition of ${\breve H}$ (\ref{q681}) is defined in (\ref{q66}),
 and $d_i$ in (\ref{q65}). Also, the matrix $H$ from (\ref{q681}) is the solution (\ref{q52}) of (\ref{q50})
 with $b=-c_0$ (\ref{q58}), that is
  \beq\label{q731}
  \begin{array}{c}
  \displaystyle{
 H=-\frac{k}{c_0}\,[\mS,\mS_x]\,.
  }
 \end{array}
 \eq
 The equations (\ref{q69})-(\ref{q70}) can be rewritten in the following way:
  \beq\label{q74}
  \begin{array}{c}
  \displaystyle{
\p_t p_i+k\p_x { d}_i+\sum\limits_{m\neq i}^N\mS_{im}\mS_{mi}(\mS_{ii}+\mS_{mm})\wp'(q_{im})-
\sum\limits_{m\neq i}^N\Big( \mS_{im}{H}_{mi}-{H}_{im}\mS_{mi} \Big)\wp(q_{im})=0
  }
 \end{array}
 \eq
 and
  \beq\label{q75}
  \begin{array}{c}
  \displaystyle{
\p_t \mS_{ij} -k\p_x{\breve H}_{ij}-2c_0(\mS_{ii}-\mS_{jj})\mS_{ij}\wp(q_{ij})+
  }
  \\ \ \\
   \displaystyle{
+(p_i-p_j)H_{ij}+\sum\limits_{m\neq i,j}^N\Big((\mS_{im}{H}_{mj}-{H}_{im}\mS_{mj})\eta(q_i,q_m,q_j)+
  }
  \\
   \displaystyle{
+\sum\limits_{m\neq i,j}^N\mS_{im}\mS_{mj}
\Big( (\mS_{ii}+\mS_{mm})\wp(q_{im})-(\mS_{jj}+\mS_{mm})\wp(q_{jm})+(\mS_{ii}-\mS_{jj})\wp(q_{ij}) \Big)=0
  }
 \end{array}
 \eq
 for $i\neq j$ and $i,j=1,...,N$.
\end{theor}
\begin{proof}
Let us combine similar terms in ${V}^{\hbox{\tiny{2dSpin}}}(z)$ (\ref{q64}):
  \beq\label{q76}
  \begin{array}{c}
  \displaystyle{
 V^{\hbox{\tiny{2dSpin}}}(z)=\delta_{ij}\Big(-d_{i}+({ H}_{ii}-O_{ii})E_1(z)
  -c_0\mS_{ii}\wp(z)\Big)+
  }
  \\ \ \\
   \displaystyle{
 +(1-\delta_{ij})\Big(-c_0\mS_{ij}F(z,q_{ij})+(H_{ij}-{O}_{ij})\phi(z,q_{ij})\Big)\,.
  }
 \end{array}
 \eq
 Then we apply the Theorem 1, where $b\rightarrow -c_0$, ${\ti H}\rightarrow {\breve H}=H-O$ and ${\ti d}_i\rightarrow -d_i$. This yields (\ref{q69})-(\ref{q73}).
 Notice that due to Proposition 2 and (\ref{q67}) the relation (\ref{q73}) transforms as:
  \beq\label{q77}
  \begin{array}{c}
  \displaystyle{
-c_0k\p_x \mS=[\mS,{\breve H}+c_0P]=[\mS,H+c_0P-O]=[\mS,H]\,,
  }
 \end{array}
 \eq
 that is $H$ is indeed given by (\ref{q731}). Next, consider equation (\ref{q72})\footnote{It is, in fact, includes
 an arbitrary function $g(t,x)$ similarly to (\ref{q48}) which we again fix to be zero.}.
 Due to (\ref{q681}), (\ref{q66}) and (\ref{q731}) from (\ref{q72}) we have
  \beq\label{q78}
  \begin{array}{c}
  \displaystyle{
 {\dot q}_i=-2(p_i+\kappa)\mS_{ii}-\frac{k}{c_0}\,[\mS,\mS_x]_{ii}\,,\quad i=1,...,N\,.
  }
 \end{array}
 \eq
 Notice that this expression for ${\dot q}_i$ is different from (\ref{q54}) and is similar to (\ref{q36}).

 Finally, the equations (\ref{q69})-(\ref{q70}) are rewritten as (\ref{q74})-(\ref{q75}) by using
 (\ref{q611}) and (\ref{q612}) from the Proposition 3. A remark is that the relations (\ref{q611}) and (\ref{q612})
 are written for the matrix $\ti O$, and they are valid for the matrix $O$ as well since
 transition $\ti O\rightarrow O$ in these formulae corresponds to
 computation in two ways of $[U^{\hbox{\tiny{2dSpin}}},V^0(z)-2\kappa U^{\hbox{\tiny{2dSpin}}}]=[U^{\hbox{\tiny{2dSpin}}},V^0(z)]$ as was explained after Proposition 2.
\end{proof}

Let us now fix the parameter $\kappa$. We choose it in such a way that $\sum\limits_{i=1}^N {\dot q}_i=0$.
Summing up equations (\ref{q78}) we obtain
  \beq\label{q781}
  \begin{array}{c}
  \displaystyle{
 \kappa=-\frac{1}{\tr(\mS)}\sum\limits_{i=1}^N\mS_{ii}p_i\,.
  }
 \end{array}
 \eq
Notice also that with this choice the sum $\sum\limits_{i=1}^N {q}_i$ is a function of $x$. It depends on initial conditions. Let us fix it be equal to zero:
  \beq\label{q782}
  \begin{array}{c}
  \displaystyle{
 \sum\limits_{i=1}^N {q}_i=0\,.
  }
 \end{array}
 \eq
 At the same time the sum of momenta is not zero. This is similar to the spinless case \cite{Krich22}.

It follows from (\ref{q71}) and (\ref{q72}) that $\p_t\mu_i(x)=0$. In principle, one may fix $\mu_i(x)$ to be a function of $x$, but we choose the level of the moment map to be independent of $x$: $\mu_i(x)=c$.

\paragraph{Consistency with the finite-dimensional model.}
 The form (\ref{q74})-(\ref{q75}) of equations of motion is useful
for the limit to finite-dimensional case. In this limit all the fields become $x$-independent variables.
That is all derivatives with respect to $x$ vanishes, or equivalently, one should put $k=0$. The constraints
$\mu_i(x)=c$ turn into $S_{ii}=c$ (the latter follows from (\ref{q71})). Then the equations (\ref{q74})-(\ref{q75}) on-shell the conditions $S_{ii}=c$ are as follows:
  \beq\label{q741}
  \begin{array}{c}
  \displaystyle{
{\dot p}_i=-2c\sum\limits_{m\neq i}^NS_{im}S_{mi}\wp'(q_{im})
  }
 \end{array}
 \eq
 and
  \beq\label{q79}
  \begin{array}{c}
  \displaystyle{
{\dot S}_{ij} =-2c\sum\limits_{m\neq i,j}^NS_{im}S_{mj}(\wp(q_{im})-\wp(q_{jm}))\,,
  }
 \end{array}
 \eq
which are the equations (\ref{q04}) with $\ti b=-2c$.

\section{Reduction to the spinless field Calogero-Moser model}
\setcounter{equation}{0}
Here we show that the obtained equations of motion for the spin field Calogero-Moser model
are in agreement with the equations for the spinless model (\ref{q07})-(\ref{q08}) from \cite{Krich22}.
For this purpose we perform a reduction to the spinless model.
Since the Hamiltonian is unknown in the field case, our strategy is to use equations of motion
similarly to the calculation (\ref{q293})-(\ref{q297}).

Suppose
\beq\label{q80}
  \begin{array}{c}
  \displaystyle{
 \mS_{ij}=\xi_i\psi_j\,.
 }
 \end{array}
 \eq
The reduction includes the constraints $\mS_{ii}=\al_i^2$ (\ref{q42}) together with some
gauge fixation, which we choose to be $\xi_i=\psi_i$, that is
\beq\label{q81}
  \begin{array}{c}
  \displaystyle{
 {\rm on\ shell}:\ \xi_i=\psi_i=(kq_{i,x}+c)^{1/2}=\al_i\,,\quad i=1,...,N\,.
 }
 \end{array}
 \eq
Then we perform the gauge transformation
  \beq\label{q251}
  \begin{array}{c}
  \displaystyle{
U^{\hbox{\tiny{2dSpin}}}(z)\rightarrow g^{-1}U^{\hbox{\tiny{2dSpin}}}(z)g-k g^{-1}{\p_x g}\,,
 }
 \end{array}
 \eq
  \beq\label{q252}
  \begin{array}{c}
  \displaystyle{
V^{\hbox{\tiny{2dSpin}}}(z)\rightarrow g^{-1}V^{\hbox{\tiny{2dSpin}}}(z)g-g^{-1}{\dot g}
 }
 \end{array}
 \eq
 with $g={\rm diag}(\xi_1,...,\xi_N)$. As a result of this gauge transformation, the matrix $U$
 depends on $\mS_{ii}$ (which are functions of $q_i$ due to (\ref{q71})-(\ref{q72}))
and $\p_x\xi_i/\xi_i$. Therefore, the spin variables are gauged away similarly to the finite-dimensional case,
and the values of $\p_x\xi_i/\xi_i$ depend on the gauge fixation.
  The $U$-$V$ pair for the spinless model appear as
\beq\label{q82}
  \begin{array}{c}
  \displaystyle{
U^{\hbox{\tiny{2dCM}}}(z)=\Big(g^{-1}U^{\hbox{\tiny{2dSpin}}}(z)g-k g^{-1}{\p_x g}\Big)\Big|_{\rm on\ shell}
 }
 \end{array}
 \eq
 and
\beq\label{q83}
  \begin{array}{c}
  \displaystyle{
V^{\hbox{\tiny{2dCM}}}(z)=\Big(g^{-1}V^{\hbox{\tiny{2dSpin}}}(z)g-g^{-1}{\dot g}\Big)\Big|_{\rm on\ shell}\,.
 }
 \end{array}
 \eq
 Since
\beq\label{q84}
  \begin{array}{c}
  \displaystyle{
 \Big(kg^{-1}{\p_x g}\Big)_{ij}=\delta_{ij}\frac{k\xi_{i,x}}{\xi_i}+k\lambda_x
 }
 \end{array}
 \eq
it is easy to see that the expression in the r.h.s.  of (\ref{q82}) coincides with (\ref{q31})
up to $k\lambda_x$, which comes from the freedom in definition of $\xi_i$, $\psi_j$ ($\xi_i\rightarrow e^\lambda\xi_i$, $\psi_j\rightarrow e^{-\lambda}\psi_j$). But this input can be gauged away and we put $\lambda=0$.

Due to (\ref{q80}) the matrix $\mS$ on shell the constraints (\ref{q81}) has the form
\beq\label{q85}
  \begin{array}{c}
  \displaystyle{
 \mS_{ij}=\xi_i\psi_j\Big|_{\rm on\ shell}=\al_i\al_j\,.
 }
 \end{array}
 \eq
Below we assume $\mS_{ij}$ to be on shell constrains. From (\ref{q301}) and (\ref{q782}) we conclude that
\beq\label{q86}
  \begin{array}{c}
  \displaystyle{
 \sum\limits_{i=1}^N \al_i^2=Nc\,,
 }
 \end{array}
 \eq
and, therefore,
\beq\label{q87}
  \begin{array}{c}
  \displaystyle{
  \mS^2=Nc\mS\,,
 }
 \end{array}
 \eq
that is
\beq\label{q88}
  \begin{array}{c}
  \displaystyle{
  c_0=Nc\,.
 }
 \end{array}
 \eq
Therefore, the function $\kappa$ turns into
\beq\label{q89}
  \begin{array}{c}
  \displaystyle{
  \kappa\Big|_{\rm on\ shell}=-\frac{1}{Nc}\sum\limits_{i=1}^N\al_i^2 p_i
 }
 \end{array}
 \eq
 as in (\ref{q302}). Also, it follows from (\ref{q86}) that
\beq\label{q90}
  \begin{array}{c}
  \displaystyle{
 \sum\limits_{i=1}^N \al_i\al_{i,x}=0\,.
 }
 \end{array}
 \eq
 Then $(\mS_x\mS)_{ij}=Nc\al_{i,x}\al_j$, $(\mS\mS_x)_{ij}=Nc\al_{i}\al_{j,x}$ and for
 the matrix $H$ (\ref{q731}) using also (\ref{q88}) we obtain
\beq\label{q91}
  \begin{array}{c}
  \displaystyle{
 H_{ij}\Big|_{\rm on\ shell}=k\al_{i,x}\al_j-k\al_{i}\al_{j,x}\,.
 }
 \end{array}
 \eq
 To finish the derivation of $U$-$V$ pair for the spinless model we need to show that the r.h.s. of
 (\ref{q83}) indeed coincides with $V^{\hbox{\tiny{2dCM}}}(z)$ (\ref{q32})-(\ref{q34}).
 Most of terms are easily obtained by substitution of (\ref{q85})-(\ref{q91}) into (\ref{q76})
 except the diagonal $z$-independent part, which is transformed under the gauge transformation
 as
\beq\label{q92}
  \begin{array}{c}
  \displaystyle{
 d_i\rightarrow d_i+\frac{{\dot \xi}_i}{\xi_i}\Big|_{\rm on\ shell}\,.
 }
 \end{array}
 \eq
 Let us compute $\frac{{\dot \xi}_i}{\xi_i}\Big|_{\rm on\ shell}$.  Due to (\ref{q293})
 we need to compute the r.h.s. of ${\dot \mS}_{ij}$ on shell the constraints. Using equations of motion
 (\ref{q75}) it is straightforwardly verified that
\beq\label{q93}
  \begin{array}{c}
  \displaystyle{
 \frac{{\dot \xi}_i}{\xi_i}\Big|_{\rm on\ shell}=k^2\frac{\al_{i,xx}}{\al_i}-2(p_i+\kappa)k\frac{\al_{i,x}}{\al_i}
 -k\p_x(p_i+\kappa)-
 }
 \\ \ \\
   \displaystyle{
 -2\sum\limits_{m\neq i}\al_m^4\wp(q_{im})-4k\sum\limits_{m\neq i}\al_m\al_{m,x}\zeta(q_{im})\,.
 }
 \end{array}
 \eq
 In this way the $V$-matrix (\ref{q32}) is reproduced.

\section{Appendix: elliptic functions}\label{secA}
\def\theequation{A.\arabic{equation}}
\setcounter{equation}{0}

Main object is the elliptic Kronecker function \cite{Weil}:
 \beq\label{a01}
  \begin{array}{l}
  \displaystyle{
 \phi(z,u)=\frac{\vth'(0)\vth(z+u)}{\vth(z)\vth(u)}=\phi(u,z)\,,\quad
 \res\limits_{z=0}\phi(z,u)=1\,,
  \quad \phi(-z, -u) = -\phi(z, u)\,,
 }
 \end{array}
 \eq
where $\vth(z)$ is the first Jacobi theta-function. In the Riemann notations it is
 \beq\label{a02}
 \begin{array}{c}
  \displaystyle{
\vth(z)=\vth(z,\tau)\equiv-\theta{\left[\begin{array}{c}
1/2\\
1/2
\end{array}
\right]}(z|\, \tau )\,,
 }
 \end{array}
 \eq
\beq\label{a03}
 \begin{array}{c}
  \displaystyle{
\theta{\left[\begin{array}{c}
a\\
b
\end{array}
\right]}(z|\, \tau ) =\sum_{j\in \z}
\exp\left(2\pi\imath(j+a)^2\frac\tau2+2\pi\imath(j+a)(z+b)\right)\,,\quad {\rm Im}(\tau)>0\,,
}
 \end{array}
 \eq
where $\tau$ is the moduli of elliptic curve $\mC/(\mZ+\tau\mZ)$.

The partial derivative $f(z,u) = \partial_u \vf(z,u)$ is
\beq\label{a04}
\begin{array}{c} \displaystyle{

    f(z, u) = \phi(z, u)(E_1(z + u) - E_1(u)), \qquad f(-z, -u) = f(z, u)
}\end{array}\eq
in terms of the first
  Eisenstein function:
\beq\label{a05}
\begin{array}{c} \displaystyle{
    E_1(z)=\frac{\vth'(z)}{\vth(z)}=\zeta(z)+\frac{z}{3}\frac{\vth'''(0)}{\vth'(0)}\,,
    \quad
    E_2(z) = - \partial_z E_1(z) = \wp(z) - \frac{\vartheta'''(0) }{3\vartheta'(0)}\,,
}\end{array}\eq
\beq\label{a06}
\begin{array}{c}
 \displaystyle{
    E_1(- z) = -E_1(z)\,, \quad E_2(-z) = E_2(z)\,,
}\end{array}
\eq
where $\wp(z)$ and $\zeta(z)$ are the Weierstrass functions.
 We use the following addition formulae:
\beq\label{a07}
  \begin{array}{c}
  \displaystyle{
  \phi(z_1, u_1) \phi(z_2, u_2) = \phi(z_1, u_1 + u_2) \phi(z_2 - z_1, u_2) + \phi(z_2, u_1 + u_2) \phi(z_1 - z_2, u_1)\,,
 }
 \end{array}
 \eq
\beq\label{a08}
  \begin{array}{c}
  \displaystyle{
 \phi(z,u_1)\phi(z,u_2)=\phi(z,u_1+u_2)\Big(E_1(z)+E_1(u_1)+E_1(u_2)-E_1(z+u_1+u_2)\Big)\,,
 }
 \end{array}
 \eq
\beq\label{a09}
  \begin{array}{c}
  \displaystyle{
  \phi(z, u_1) f(z,u_2)-\phi(z, u_2) f(z,u_1)=\phi(z,u_1+u_2)\Big(\wp(u_1)-\wp(u_2)\Big)\,,
 }
 \end{array}
 \eq
\beq\label{a10}
  \begin{array}{c}
  \displaystyle{
  \phi(z, u) \phi(z, -u) = \wp(z)-\wp(u)=E_2(z)-E_2(u)\,,
 }
 \end{array}
 \eq
\beq\label{a11}
  \begin{array}{c}
  \displaystyle{
  \phi(z, u) f(z, -u)-\phi(z, -u) f(z, u)=\wp'(u)\,.
 }
 \end{array}
 \eq
The second order partial derivative
$f'(z,u) = \partial^2_u \phi(z,u)$:
\beq\label{a111}
\begin{array}{c}
 \displaystyle{
    f'(z,u)=\phi(z,u)\Big(\wp(z)-E_1^2(z)+2\wp(u)-2E_1(z)E_1(u)+2E_1(z+u)E_1(z)\Big)=
 }
 \\
 \displaystyle{
   =\Big(\wp(z)-E_1^2(z)+2\wp(u)\Big)\phi(z,u)+2E_1(z)f(z,u) \,.
}\end{array}
\eq
Introduce the function
\beq\label{a12}
  \begin{array}{c}
  \displaystyle{
 F(z,u)=E_1(z)\phi(z,u)-f(z,u)=\Big(E_1(z)+E_1(u)-E_1(z+u)\Big)\phi(z,u)\,.
 }
 \end{array}
 \eq
Then
\beq\label{a13}
  \begin{array}{c}
  \displaystyle{
 F'(z,u)\equiv\p_u F(z,u)=-\Big(\wp(z)-E_1^2(z)+2\wp(u)\Big)\phi(z,u)-E_1(z)f(z,u)=
  }
 \\
 \displaystyle{
 =E_1(z)F(z,u)-\Big(\wp(z)+2\wp(u)\Big)\phi(z,u)\,.
 }
 \end{array}
 \eq
Also,
\beq\label{a14}
  \begin{array}{c}
  \displaystyle{
 \phi(z,u_1)f(z,u_2)=\phi(z,u_1+u_2)\Big(E_1^2(z)-\wp(z)-\wp(u_1+u_2)-\wp(u_2)\Big)-
  }
 \\
 \displaystyle{
 -E_1(z)f(z,u_1+u_2)+\Big(E_1(u_1)+E_1(u_2)-E_1(u_1+u_2)\Big)f(z,u_1+u_2)
 }
 \end{array}
 \eq
and
\beq\label{a15}
  \begin{array}{c}
  \displaystyle{
   \Big(\zeta(z+w)-\zeta(z)-\zeta(w)\Big)^2=\Big(E_1(z+w)-E_1(z)-E_1(w)\Big)^2=\wp(z)+\wp(w)+\wp(z+w)\,.
 }
 \end{array}
 \eq
%


\subsection*{Acknowledgments}
This work was performed at the Steklov International Mathematical Center and supported by the Ministry of Science and Higher Education of the Russian Federation (agreement no. 075-15-2022-265).




\begin{small}

\end{small}

\end{document}